\documentclass[letterpaper,10pt,conference]{ieeeconf}
\IEEEoverridecommandlockouts
\usepackage{amsthm}
\usepackage{amsmath}
\usepackage{amssymb}
\usepackage{mathtools}

\let\labelindent\relax
\usepackage{enumitem}
\usepackage{graphicx}
\usepackage{cite}
\usepackage[hidelinks]{hyperref}
\usepackage{cleveref}
 
\newtheorem{assumption}{\normalfont \textbf{Assumption}}
\newtheorem{lemma}{\normalfont \textbf{Lemma}}
\newtheorem{theorem}{\normalfont \textbf{Theorem}}
\newtheorem{proposition}{\normalfont \textbf{Proposition}}
\newtheorem{definition}{\normalfont \textbf{Definition}}
\newtheorem{remark}{\normalfont \textbf{Remark}}

\crefname{assumption}{Assumption}{Assumptions}
\Crefname{assumption}{Assumption}{Assumptions}
\usepackage{xcolor}

\newcommand{\R}{\mathbb{R}}
\newcommand{\sgn}{\mathrm{sgn}}
\newcommand{\diag}{\mathrm{diag}}
\title{A Networked SIS Epidemic--Opinion Model with Higher-Order Interactions}
\author{Saba Samadi, Jos\'{e} I. Caiza, Sebin Gracy, and Philip E. Par\'{e}* \thanks{*Saba Samadi, Jos\'{e} I. Caiza and Philip E.~Par\'e are with the Elmore Family School of Electrical and Computer Engineering at Purdue University. Sebin Gracy is with the Department of Electrical Engineering and Computer Science at South Dakota School of Mines and Technology. Emails:~ssamadi@purdue.edu, jcaiza@purdue.edu, philpare@purdue.edu, 
sebin.gracy@sdsmt.edu. 
This work was supported in part 
   by the National Science Foundation, grant
   NSF-ECCS \#2238388.}}

\begin{document}
\maketitle

\begin{abstract}
This paper studies a susceptible--infected--susceptible (SIS) epidemic model coupled with opinion dynamics over a network of communities with higher-order interactions. Unlike standard networked SIS models, which account only for pairwise transmission, the proposed model incorporates group-level infection mechanisms and coupling between epidemic prevalence and community opinions. We secure conditions for the (local) stability (resp. instability) of a particular healthy equilibrium, derive a sufficient condition for global exponential eradication of an infection state, and identify conditions under which higher-order interactions induce bistability in the reduced dynamics on a positively invariant synchronous set.
These results characterize how higher-order interactions alter the dynamics of opinion-dependent epidemic systems. Numerical simulations illustrate the predicted eradication, bistable, and endemic regimes.
\end{abstract}

\section{Introduction}\label{sec:introduction}

Network epidemic models play an important role in understanding and controlling the spread of infectious diseases over interconnected populations \cite{nowzari2016analysis, pare2020modeling, mei2017dynamics}. The susceptible--infected--susceptible (SIS) model is a widely studied framework for diseases where recovery does not confer lasting immunity \cite{lajmanovich1976deterministic, van2008virus, khanafer2016stability}. Classical networked SIS epidemic models describe contagion through pairwise contacts encoded by a graph. However, real-world contagion events often involve group-level exposures that simple graphs cannot faithfully capture. This limitation has motivated the use of higher-order network structures such as hypergraphs and simplicial complexes \cite{battiston2021physics, bick2023higher}.
In this direction, SIS spreading over hypergraphs has been studied in 
\cite{bodo2016sis,iacopini2019simplicial}, where simplicial interactions 
are shown to induce discontinuous phase transitions and bistability. 
Other works have established conditions for disease-free, endemic, and 
bistable regimes in multigroup SIS models, addressed nonlinear spreading 
and diffusion on higher-order and directed hypergraphs, and studied 
competitive epidemics with higher-order interactions~\cite{cisneros2021multigroup, stonge2021universal, cui2023general, gracy2025}.

Epidemic evolution is influenced not only by contact structure but also by how
communities perceive the severity of an outbreak. Perceived severity can alter
protective behavior and social contact, motivating coupled
behavior--epidemic models \cite{ye2021game}. 
Following~\cite{she2022networked}, we let opinion modulate transmission and an
\emph{effective recovery rate}, while the local infection fraction non-trivially influences the opinion evolution. Here, the effective recovery rate represents mechanisms
that reduce the infectious fraction, including treatment, adherence, testing,
isolation, and community intervention; it is not interpreted as a purely
biological recovery rate. The signed opinion-dynamics component is related to antagonistic-consensus models such as \cite{altafini2012consensus}. Related directions include discrete-time SIS--opinion models with stubborn agents \cite{xu2025coupled}, game-theoretic epidemic--behavior models \cite{ye2021game,catalano2025game}, and multilayer coevolutionary models with competing opinions \cite{peng2021multilayer}. However, these works are restricted to pairwise contagion.

Recent work has begun to combine behavioral or opinion feedback with
higher-order network structure in related settings. For example,
\cite{kan2026modeling} incorporates simplicial reinforcement in an awareness
layer, while \cite{wen2026opinion} studies opinion modulation and higher-order
contacts in a dual-virus model. The present work differs in its focus on a
single-virus networked SIS model with state-dependent signed opinion dynamics,
a global eradication condition uniform over admissible opinion states, and conditions under which higher-order transmission induces bistability on a
synchronous invariant set. In this
sense, the higher-order-interaction
susceptible--infected--opinion--susceptible (HOI-SIOS) model combines and extends the opinion-dependent SIS
framework in \cite{she2022networked} and the higher-order SIS framework in
\cite{cisneros2021multigroup}.
The main contributions of this paper are as follows:
\begin{enumerate}[label=(\roman*)]
    \item We formulate a coupled SIS--opinion model with pairwise and quadratic
    higher-order transmission, opinion-dependent transmission and removal
    rates, and state-dependent signed opinion dynamics.

    \item We establish positive invariance, characterize local stability and
    instability of the maximally skeptical healthy equilibrium, and derive a
    sufficient condition for global exponential eradication.

    \item Under a homogeneous setting, we reduce the synchronous dynamics to two
dimensions and derive bistability conditions. We isolate the role of
higher-order transmission and give a sufficient condition for instability of
healthy equilibria in the heterogeneous model.
\end{enumerate}

\textit{Notation:}
We let [n] denote the set of the first n natural numbers, i.e., \([n]\coloneqq\{1,\ldots,n\}\). The symbols \(\mathbf{0}_n\),
\(\mathbf{1}_n\), and \(I_n\) denote the \(n\)-dimensional zero vector, the
all-ones vector, and the \(n\times n\) identity matrix, respectively. For
vectors or matrices of the same size, inequalities \(u\le v\) are
understood componentwise. In particular, \(x\ge \mathbf{0}_n\) and
\(x\gg \mathbf{0}_n\) denote componentwise non-negativity and strict positivity.
For a square matrix \(M\), \(\rho(M)\) denotes its spectral radius,
\(\mathbf{s}(M)\) its spectral abscissa, \(\det(M)\) its determinant, and
\(\operatorname{tr}(M)\) its trace. The operator \(\diag(\cdot)\) constructs a
diagonal matrix from a vector; when applied to a square matrix, it extracts its
diagonal. A square matrix is nonnegative if all its entries are nonnegative,
and it is Metzler if all its off-diagonal entries are nonnegative. For
\(x\in[0,1]^n\) and \(o\in[-0.5,0.5]^n\), we write
\(X\coloneqq\diag(x)\), \(O\coloneqq\diag(o)\), and
\(o'\coloneqq o+0.5\,\mathbf{1}_n\), so that
\(O'=\diag(o')=O+0.5I_n\).
The modified sign function is defined as
\[
\sgn_m(z)\coloneqq
\begin{cases}
1, & z\ge 0,\\
-1, & z<0,
\end{cases}
\]
and the gauge matrix is
\[
\Phi(o)\coloneqq \diag(\sgn_m(o_1),\ldots,\sgn_m(o_n)).
\]
The convention \(\sgn_m(0)=1\) ensures that \(\Phi(o)\) is defined for all
\(o\in[-0.5,0.5]^n\) and has diagonal entries in \(\{-1,1\}\), so that
\(\Phi(o)^2=I_n\), even when some \(o_i=0\).
For a function $f(\mu,\nu)$, we write $\partial_\mu f$ and $\partial_\nu f$
for the partial derivatives with respect to the first and second arguments,
respectively.
\section{Model Formulation}\label{sec:model}
We consider a population partitioned into $n$ interacting communities. For
each $i\in[n]$, the state $x_i(t)\in[0,1]$ is the infected fraction of community
$i$ at time instant t, and $o_i\in[-0.5,0.5]$ represents its perceived severity of the epidemic;
larger values indicate greater concern. The epidemic and opinion states are
coupled in both directions. The local infection fraction influences perceived severity,
while perceived severity changes protective behavior and the effective
transmission and recovery rates.

Pairwise transmission is described by a nonnegative adjacency matrix
$A=[a_{ij}]$, where $a_{ij}x_j$ represents the exposure of community $i$ to
the infected fraction in community $j$. We also include quadratic higher-order
transmission. For every community $i$, let
$B_i=[b_{ijk}]\in\mathbb R_{\ge0}^{n\times n}$ denote its higher-order
interaction matrix. The term $b_{ijk}x_jx_k$ represents joint exposure of the community $i$ to infected sources in communities $j$ and $k$; the source
communities may coincide. Under a mean-field factorization, $x_jx_k$ is the
corresponding joint-infection probability. The resulting higher-order
transmission term is
\[
\mathcal A_2(x)
\coloneqq
\bigl[x^\top B_1x,\ldots,x^\top B_nx\bigr]^\top.
\]
When $i$, $j$, and $k$ are distinct and the coefficients are induced by a
simplicial complex, $b_{ijk}>0$ can represent a filled triangle
$\{i,j,k\}$. Repeated indices instead represent within-community or
repeated-source quadratic exposure. We focus on quadratic interactions because they are the lowest-order
nonlinear terms that can produce the bistability studied below. Higher-order products can be introduced
analogously, but their equilibrium analysis leads to higher-degree equations.

\subsection{Opinion-Dependent Rates}

Define the diagonal parameter matrices:
\begin{align*}
\Gamma_{\min} &\coloneqq \gamma^{\min} I_n, \quad
\Gamma \coloneqq \diag(\gamma_1, \ldots, \gamma_n),\\
B_1^{\min} &\coloneqq \beta_1^{\min} I_n, \quad
\beta_1 \coloneqq \diag(\beta_{11}, \ldots, \beta_{n1}), \\
B_2^{\min} &\coloneqq \beta_2^{\min} I_n, \quad
\beta_2 \coloneqq \diag(\beta_{12}, \ldots, \beta_{n2}).
\end{align*}
The opinion-dependent effective recovery and transmission-rate matrices are:
\begin{align}
\mathcal{D}(o) &\coloneqq \Gamma_{\min} + (\Gamma - \Gamma_{\min}) O', \label{eq:D_def} \\
\mathcal{B}_q(o) &\coloneqq \beta_q - (\beta_q - B_q^{\min}) O', \quad q \in \{1, 2\}. \label{eq:B_def}
\end{align}

As $o_i'$ increases from zero to one, the $i$th effective recovery rate
increases from $\gamma^{\min}$ to $\gamma_i$, while the corresponding
transmission rates decrease from $\beta_{iq}$ to $\beta_q^{\min}$. We interpret
$\mathcal D(o)$ as an effective recovery rate rather than a purely biological
recovery rate: it aggregates treatment seeking, adherence, testing, isolation,
and community intervention. The affine dependence interpolates between the skeptical and concerned
regimes.

\begin{assumption}\label{ass:parameters}
For all $i \in [n]$:
\begin{enumerate}[label=(\roman*)]
    \item $\gamma_i \ge \gamma^{\min} > 0$,
    \item $\beta_{i1} \ge \beta_1^{\min} \ge 0$,
    \item $\beta_{i2} \ge \beta_2^{\min} \ge 0$.
\end{enumerate}
\end{assumption}

\begin{assumption}
\label{ass:network}
The pairwise infection matrix satisfies $A\ge0$, the higher-order interaction matrices satisfy $B_i\ge0$ for all $i\in[n]$, and the unsigned opinion matrix
$\bar A_u\ge0$ is symmetric.
\end{assumption}

The analysis does not invoke irreducibility; disconnected pairwise and opinion
graphs are therefore allowed.

\subsection{Coupled HOI-SIOS Dynamics}

Under Assumptions~\ref{ass:parameters} and~\ref{ass:network}, the coupled epidemic--opinion dynamics are:
\begin{subequations}\label{eq:HOI-SIOS}
\begin{align}
\dot{x} &= -\mathcal{D}(o) x + (I_n - X) \Big( \mathcal{B}_1(o) A x + \mathcal{B}_2(o) \mathcal{A}_2(x) \Big), \label{eq:x_dynamics} \\
\dot{o} &= x - (I_n + \Phi(o) \bar{L}_u \Phi(o)) o - 0.5 \mathbf{1}_n, \label{eq:o_dynamics}
\end{align}
\end{subequations}
where $\bar{L}_u$ is the Laplacian of the unsigned opinion graph with adjacency matrix $\bar{A}_u$.
We define the signed opinion adjacency matrix
\(
\bar A(o)\coloneqq \Phi(o)\,\bar A_u\,\Phi(o),
\text{ and }
\bar a_{ij}(o)\coloneqq [\bar A(o)]_{ij}.
\)
The node-wise form is:
\begin{subequations}\label{eq:HOI-SIOS-node}
\begin{align}
\dot{x}_i &= -\big[\gamma^{\min} + (\gamma_i - \gamma^{\min}) o_i'\big] x_i \nonumber \\
&\quad + (1 - x_i) \sum_{j=1}^n \big[\beta_{i1} - (\beta_{i1} - \beta_1^{\min}) o_i'\big] a_{ij} x_j \nonumber \\
&\quad + (1 - x_i) \sum_{j,k=1}^n \big[\beta_{i2} - (\beta_{i2} - \beta_2^{\min}) o_i'\big] b_{ijk} x_j x_k, \label{eq:x_node} \\
\dot{o}_i &= (x_i - o_i') + \sum_{j=1}^n |\bar{a}_{ij}(o)| \Big( \sgn(\bar{a}_{ij}(o)) o_j - o_i \Big). \label{eq:o_node}
\end{align}
\end{subequations}

The unsigned matrix $\bar A_u$ specifies which communities exchange opinions
and the strength of those exchanges, independently of whether the interaction
is cooperative or antagonistic. The gauge matrix determines the sign of each interaction from the signs of
the current opinions. In particular, whenever
$[\bar A_u]_{ij}>0$,
\[
\sgn\bigl(\bar a_{ij}(o)\bigr)
=\sgn_m(o_i)\sgn_m(o_j).
\]
Hence, communities with the same opinion sign interact cooperatively, whereas
communities with opposite signs interact antagonistically. For every fixed sign pattern, this construction produces a
structurally balanced signed graph. The edge signs switch when an opinion
crosses zero, while the unsigned interaction strengths remain fixed.

The local term $x_i-o_i'$ in~\eqref{eq:o_node} causes perceived severity to
track the local infection fraction: concern increases when $x_i>o_i'$ and
decreases when $x_i<o_i'$. The signed-network term then couples this local
response to neighboring communities.

\begin{definition}[State domain]\label{statedomain}
Define the open domain $\mathbf{D} \coloneqq (0,1)^n \times (-0.5, 0.5)^n$ and its closure $\overline{\mathbf{D}} = [0,1]^n \times [-0.5, 0.5]^n$.
\end{definition}

\begin{remark}[Solution convention]
Because $o\mapsto\Phi(o)$ is discontinuous on the switching surfaces
$\{o:o_i=0\}$, solutions are understood as absolutely continuous functions
satisfying~\eqref{eq:HOI-SIOS} almost everywhere. All subsequent statements
concerning solutions use this convention; uniqueness at the switching surfaces
is not required. Jacobian-based stability results are restricted to equilibria
satisfying $o_i^*\ne0$ for all $i$, where the vector field is smooth in a full
neighborhood. The comparison inequality in
Theorem~\ref{thm:global-stability-lyap}(see Section~\ref{sec:stability}) is likewise required only almost
everywhere.
\end{remark}

\begin{lemma}[Positive invariance]\label{lem:invariance}
Under Assumptions~\ref{ass:parameters} and~\ref{ass:network}, $\overline{\mathbf{D}}$ defined in Definition~\ref{statedomain} is positively invariant for the system~\eqref{eq:HOI-SIOS}.
\end{lemma}

\begin{proof}
We analyze each subsystem (i.e., the epidemic dynamics and the opinion dynamics) at the boundaries. 
From~\eqref{eq:x_node}, define $f_i(x,o)$ as the right-hand side.

Fix any $i\in[n]$ and suppose that $x_i=0$. Then the removal term vanishes, while all infection terms are nonnegative since $A\ge 0$, $B_i\ge 0$, $x_j\ge 0$, and $\mathcal{B}_q(o)\ge 0$ by Assumption~\ref{ass:parameters}. Hence, $\dot{x}_i=f_i(x,o)\ge 0$.
Suppose that $x_i=1$, which implies that $(1 - x_i) = 0$. The factor $(1 - x_i) = 0$ eliminates all the infection terms, leaving
\[
\dot{x}_i = -[\mathcal{D}(o)]_{ii} = -\big[\gamma^{\min} + (\gamma_i - \gamma^{\min}) o_i'\big] < 0,
\]
since $\gamma^{\min} > 0$ and $o_i' \ge 0$. Thus $\dot{x}_i < 0$. Therefore, it follows that $x(t) \in [0,1]^n$  for all $ t \in \mathbb{R}_{\geq 0}$.

Next, rewrite~\eqref{eq:o_dynamics} as:

\begin{equation}
    \dot{o}_i = x_i - o_i' + \sum_{j=1}^n |\bar{a}_{ij}(o)| \Big( \sgn(\bar{a}_{ij}(o)) o_j - o_i \Big). \label{eq:opinion}
\end{equation}
Suppose that $o_i = -0.5$, which implies that  $o_i' = 0$ and $x_i - o_i' = x_i \ge 0$.
For any $j$, since $o_j \in [-0.5, 0.5]$, it must be that:
\[
\sgn(\bar{a}_{ij}(o)) o_j - o_i = \sgn(\bar{a}_{ij}(o)) o_j - (-0.5) \in [0, 1].
\]
Since $|\bar{a}_{ij}(o)| \ge 0$, the summation term on the right-hand side of~\eqref{eq:opinion} is nonnegative. Thus $\dot{o}_i \ge 0$.
Suppose that  $o_i = 0.5$, which implies that
$o_i' = 1$ and, therefore, $x_i - o_i' = x_i - 1 \le 0$ (since $x_i \le 1$). For any $j$:
\[
\sgn(\bar{a}_{ij}(o)) o_j - o_i = \sgn(\bar{a}_{ij}(o)) o_j - 0.5 \in [-1, 0].
\]
Thus, the summation term on the right-hand side of~\eqref{eq:opinion} is nonpositive 
and, hence, $\dot{o}_i \le 0$.
Therefore, $o(t) \in [-0.5, 0.5]^n$, for all $ t \in \mathbb{R}_{\geq 0}$. 

Since $x(t) \in [0,1]^n$  and $o(t) \in [-0.5, 0.5]^n$ for all $ t \in \mathbb{R}_{\geq 0}$, $\overline{\mathbf{D}}$ is positively invariant.
\end{proof}

Before analyzing stability, we introduce terminology for the equilibrium
regimes considered below.

\begin{definition}[Equilibrium types]\label{def:equilibria}
An equilibrium $(x^*,o^*)\in\overline{\mathbf D}$ is \emph{healthy} if
$x^*=\mathbf0_n$ and \emph{endemic} if $x^*\gg\mathbf0_n$.
\end{definition}

Since the effective removal and transmission rates depend on the opinion
state, define the extremal pairwise reproduction numbers
\begin{equation}\label{eq:Rmin-Rmax}
\mathcal R_{\min}\coloneqq \rho(\Gamma^{-1}B_1^{\min}A),
\quad
\mathcal R_{\max}\coloneqq \rho(\Gamma_{\min}^{-1}\beta_1 A).
\end{equation}
The value $\mathcal R_{\max}$ corresponds to the maximally skeptical
opinion profile ($o'=\mathbf{0}_n$), for which effective removal is minimal
and pairwise transmission is maximal. Conversely, $\mathcal R_{\min}$
corresponds to the maximally concerned profile ($o'=\mathbf{1}_n$), for
which effective removal is maximal and pairwise transmission is minimal.
For every admissible opinion profile,
\(
\mathcal R_{\min}
\le
\rho\!\left(\mathcal D(o)^{-1}\mathcal B_1(o)A\right)
\le
\mathcal R_{\max}.
\)
These quantities are used below to distinguish the healthy and endemic
parameter regimes.

\section{Stability of a Healthy Equilibrium}\label{sec:stability}
Direct substitution shows that
$(x^h,o^h)\coloneqq(\mathbf0_n,-0.5\mathbf1_n)$ is a healthy equilibrium:
$\mathcal A_2(\mathbf0_n)=\mathbf0_n$ gives $\dot x=\mathbf0_n$, while
$\bar L_u\mathbf1_n=\mathbf0_n$ gives $\dot o=\mathbf0_n$. The equilibrium
$(x^h,o^h)$ pairs disease eradication with the maximally skeptical opinion
profile. Since healthy equilibria may be nonunique, the local analysis below
concerns $(x^h,o^h)$. The subsequent global result concerns only convergence
of the infection state, because the zero-infection opinion subsystem may admit multiple equilibria.

Define the combined infection term:
\[
y(x,o) \coloneqq \mathcal{B}_1(o) A x + \mathcal{B}_2(o) \mathcal{A}_2(x).
\]
For the HOI term, $[\mathcal{A}_2(x)]_i = x^\top B_i x$. 
Thus, the Jacobian of $\mathcal{A}_2$ with respect to $x$ is:
\[
J_{\mathcal{A}_2}(x) \coloneqq 
\begin{bmatrix}
\big((B_1 + B_1^\top) x\big)^\top \\
\vdots \\
\big((B_n + B_n^\top) x\big)^\top
\end{bmatrix} \in \R^{n \times n}.
\]
Equivalently, $[J_{\mathcal{A}_2}(x)]_{ij} = [(B_i + B_i^\top) x]_j$.


Since the map $o\mapsto \Phi(o)$ is piecewise constant,
system~\eqref{eq:HOI-SIOS} is piecewise smooth. Accordingly, the Jacobian
computed below is used only in regions where $\Phi(o)$ is constant, i.e., away
from the switching surfaces $\{o:o_i=0\}$. In particular, all Jacobian-based
stability arguments below are applied only at equilibria where $o_i^*\neq 0$
for every $i\in[n]$. On each such region, system~\eqref{eq:HOI-SIOS} is smooth,
and the Jacobian has the block structure:
\[
J(x,o) = 
\begin{bmatrix}
J_{xx}(x,o) & J_{xo}(x,o) \\[1mm]
I_n & J_{oo}(o)
\end{bmatrix},
\]
where:
\begin{align}
J_{xx}(x,o) &= -\mathcal{D}(o) - \diag(y(x,o)) \nonumber \\
&\quad + (I_n - X) \Big( \mathcal{B}_1(o) A + \mathcal{B}_2(o) J_{\mathcal{A}_2}(x) \Big), \nonumber\\
J_{xo}(x,o) &= -\diag\Big( (\Gamma - \Gamma_{\min}) x \nonumber \\
&\quad + (I_n - X) \big[ (\beta_1 - B_1^{\min}) A x \nonumber \\
&\quad + (\beta_2 - B_2^{\min}) \mathcal{A}_2(x) \big] \Big), \nonumber\\
J_{oo}(o) &= -(I_n + \Phi(o) \bar{L}_u \Phi(o)).\nonumber 
\end{align}

\subsection{Local Stability}
Since \(o^h=-0.5\mathbf{1}_n\), there exists a neighborhood of \((x^h,o^h)\)
on which \(o_i<0\) for all \(i\in[n]\). Hence, \(\Phi(o)=-I_n\) is constant on
that neighborhood, and system~\eqref{eq:HOI-SIOS} is locally smooth. Therefore,
Lyapunov’s indirect method applies at \((x^h,o^h)\), yielding the following proposition.
\begin{proposition}[Local exponential stability]\label{prop:local-stability}
Consider system~\eqref{eq:HOI-SIOS} under Assumptions~\ref{ass:parameters} and~\ref{ass:network}. If
\(
\mathcal R_{\max}<1,
\)
then the healthy equilibrium \((x^h,o^h)=(\mathbf{0}_n,-0.5\mathbf{1}_n)\) is locally
exponentially stable. If
\(
\mathcal R_{\max}>1,
\)
then \((x^h,o^h)\) is unstable.
\end{proposition}

\begin{proof}
We evaluate the Jacobian of~\eqref{eq:HOI-SIOS} at
\((x^h,o^h)=(\mathbf{0}_n,-0.5\mathbf{1}_n)\). Since \(o^h=-0.5\mathbf{1}_n\), we have $O'=0$.
Thus, from~\eqref{eq:D_def} and~\eqref{eq:B_def},
\[
\mathcal D(o^h)=\Gamma_{\min},\qquad
\mathcal B_1(o^h)=\beta_1,\qquad
\mathcal B_2(o^h)=\beta_2.
\]

Since $x^h=\mathbf{0}_n$, we have $y(\mathbf{0}_n,o^h)=\mathbf{0}_n$, $X=0$, and
$J_{\mathcal A_2}(\mathbf{0}_n)=0$, because $J_{\mathcal A_2}(x)$ is linear in $x$.
Hence,
\[
J_{xx}(\mathbf{0}_n,o^h)=-\Gamma_{\min}+\beta_1A.
\]
Also, since $x=\mathbf{0}_n$ and $\mathcal A_2(\mathbf{0}_n)=\mathbf{0}_n$,
\(
J_{xo}(\mathbf{0}_n,o^h)=0.
\)
Finally, since $\Phi(o^h)=-I_n$,
\[
J_{oo}(o^h)=-(I_n+(-I_n)\bar L_u(-I_n))=-(I_n+\bar L_u).
\]
Therefore, the Jacobian at $(x^h,o^h)$ is
\begin{equation*}\label{eq:jacobian:healthy}
J(\mathbf{0}_n,o^h)=
\begin{bmatrix}
-\Gamma_{\min}+\beta_1A & 0\\
I_n & -(I_n+\bar L_u)
\end{bmatrix}.
\end{equation*}
Since this matrix is lower block triangular, its eigenvalues are precisely the
eigenvalues of $-\Gamma_{\min}+\beta_1A$ and $-(I_n+\bar L_u)$.

We first consider the epidemic block
\(
-\Gamma_{\min}+\beta_1A.
\)
This is a Metzler matrix because $\beta_1A\ge 0$. For Metzler matrices,
Hurwitz stability of $-\Gamma_{\min}+\beta_1A$ is equivalent to
\(
\mathcal R_{\max}<1.
\)

Next, the opinion block $-(I_n+\bar L_u)$ has eigenvalues
$-(1+\lambda_i(\bar L_u))$, where $\lambda_i(\bar L_u)\ge 0$ because
$\bar L_u$ is a Laplacian matrix. Hence, all eigenvalues of $J_{oo}(o^h)$ have a
negative real part, so $J_{oo}(o^h)$ is Hurwitz.
Therefore, if $\mathcal R_{\max}<1$, then all eigenvalues of
$J(\mathbf{0}_n,o^h)$ have negative real parts, and
$(\mathbf{0}_n,-0.5\mathbf{1}_n)$ is locally exponentially stable by
\cite[Theorem~4.15 and Corollary~4.3]{khalil2002nonlinear}.

If $\mathcal R_{\max}>1$, then
$\mathbf s(-\Gamma_{\min}+\beta_1A)>0$ \cite{liu2019analysis}.
Because the full Jacobian is lower block triangular, it consequently has an
eigenvalue with positive real part. Lyapunov's indirect method therefore
implies that $(\mathbf0_n,-0.5\mathbf1_n)$ is unstable
\cite[Theorem~4.7(ii)]{khalil2002nonlinear}.
\end{proof}

\subsection{Global Disease Eradication}

We now derive a stronger condition that guarantees global eradication of the epidemic
state. Since the asymptotic opinion profile need not be unique, the result is
stated only in terms of the convergence of $x(t)$ to $\mathbf{0}_n$.

\begin{lemma}[Linear bound for $\mathcal{A}_2(x)$]\label{lem:Cbound} Define the matrix $C=[c_{ik}]\in\mathbb{R}^{n\times n}$ 
where $c_{ik}\coloneqq \sum_{j=1}^n b_{ijk}$. 
Then, for all $x\in[0,1]^n$, $\mathcal{A}_2(x)\le Cx,$ 
componentwise. 
\end{lemma} 
\begin{proof} For each $i\in[n]$ and $x\in[0,1]^n$, \footnotesize \[ [\mathcal{A}_2(x)]_i = \!\!\sum_{j,k=1}^n\!b_{ijk} x_j x_k \le\! \!\sum_{j,k=1}^n\!b_{ijk} x_k = \!\sum_{k=1}^n\!\bigg(\!\sum_{j=1}^n b_{ijk}\!\bigg) x_k = [Cx]_i, \] \normalsize where we used $x_j \le 1$ for all $j$. 
\end{proof}

\begin{theorem}[Global disease eradication]\label{thm:global-stability-lyap}
Consider system~\eqref{eq:HOI-SIOS} under Assumptions~\ref{ass:parameters}
and~\ref{ass:network}. Define
\(
M\coloneqq \Gamma_{\min}^{-1}(\beta_1A+\beta_2C).
\)
If $\rho(M)<1$, then every solution of~\eqref{eq:HOI-SIOS} with initial
condition $(x(0),o(0))\in\overline{\mathbf D}$ satisfies
\(
x(t)\to \mathbf{0}_n \text{ exponentially as }t\to\infty.
\)
\end{theorem}

\begin{proof}
For every $(x,o)\in\overline{\mathbf D}$, we have $o_i'\in[0,1]$, and hence
\[
\mathcal D(o)\ge \Gamma_{\min},\quad
\mathcal B_q(o)\le \beta_q,\quad q\in\{1,2\}.
\]
Also, since $x\in[0,1]^n$, the diagonal matrix $(I_n-X)$ satisfies
\(
0\le I_n-X\le I_n.
\)
By Lemma~\ref{lem:Cbound}, $\mathcal A_2(x)\le Cx$. Therefore, the epidemic
dynamics~\eqref{eq:x_dynamics} satisfy
\begin{align}
\dot x
&=
-\mathcal D(o)x+(I_n-X)\big(\mathcal B_1(o)Ax+\mathcal B_2(o)\mathcal A_2(x)\big)
\nonumber\\
&\le
-\Gamma_{\min}x+(\beta_1A+\beta_2C)x
=
\Gamma_{\min}(M-I_n)x.
\label{eq:bound:x:by:M}
\end{align}

Since $M\ge 0$ and $\rho(M)<1$, the matrix $I_n-M^\top$ is invertible and
inverse-positive. Define
\(
w\coloneqq (I_n-M^\top)^{-1}\mathbf{1}_n.
\)
Then $w\gg \mathbf{0}_n$ and
\(
w^\top(M-I_n)=-\mathbf{1}_n^\top.
\)

Now, define the Lyapunov function
\[
V(x)\coloneqq w^\top \Gamma_{\min}^{-1}x
=
\frac{1}{\gamma^{\min}}\,w^\top x.
\]
For $x\ge \mathbf{0}_n$, we have $V(x)\ge 0$, and $V(x)=0$ if and only if $x=\mathbf{0}_n$.
Along trajectories of~\eqref{eq:HOI-SIOS}, using~\eqref{eq:bound:x:by:M},
\begin{align*}
\dot V(x)
&=
w^\top \Gamma_{\min}^{-1}\dot x
\le
w^\top \Gamma_{\min}^{-1}\Gamma_{\min}(M-I_n)x \\
&=
w^\top(M-I_n)x
=
-\mathbf{1}_n^\top x.
\end{align*}
Let $w_{\max}\coloneqq \max_{i\in[n]} w_i$. Since
\[
V(x)=\frac{1}{\gamma^{\min}}w^\top x
\le
\frac{w_{\max}}{\gamma^{\min}}\mathbf{1}_n^\top x,
\]
it follows that
\(
\mathbf{1}_n^\top x \ge \frac{\gamma^{\min}}{w_{\max}}\,V(x).
\)
Hence,
\[
\dot V(x)\le -\frac{\gamma^{\min}}{w_{\max}}\,V(x).
\]
By the comparison lemma~\cite[Lemma~3.4]{khalil2002nonlinear},
\[
V(x(t))\le e^{-\frac{\gamma^{\min}}{w_{\max}}t}V(x(0)).
\]
Since $w\gg \mathbf{0}_n$, the function $V$ is equivalent to any norm on
$\mathbb R_{\ge 0}^n$, and therefore $x(t)\to \mathbf{0}_n$ exponentially.
\end{proof}

Theorem~\ref{thm:global-stability-lyap} guarantees global eradication of the
epidemic component. It does not, by itself, imply the convergence of the opinion
state to $-0.5\mathbf{1}_n$ or to a unique healthy equilibrium, since the
zero-infection opinion subsystem
\[
\dot o = -\big(I_n+\Phi(o)\bar L_u\Phi(o)\big)o-0.5\mathbf{1}_n
\]
may admit multiple equilibria. Establishing the convergence of $o(t)$ requires
additional assumptions on the opinion dynamics and is not pursued here.
Since $\beta_2C\ge0$, the monotonicity of the spectral radius implies
$\mathcal R_{\max}\le\rho(M)$. Hence, $\rho(M)<1$ implies
$\mathcal R_{\max}<1$, whereas the converse need not hold.

\section{Non-Healthy Equilibria and Bistability}\label{sec:endemic}
In this section, we analyze the model beyond the disease-free regime. Under a homogeneous setting, we show that HOIs can
induce healthy--endemic bistability even when $\mathcal R_{\max}<1$. When
$\mathcal R_{\min}>1$, we establish the existence of an endemic equilibrium on
the synchronous set and prove that every healthy equilibrium of the general
heterogeneous model away from the opinion switching surfaces is unstable.

\subsection{Bistability on a Synchronous Invariant Set}\label{subsec:BiSynch}
We now introduce a homogeneous setting in which the full dynamics admit a
positively invariant set of synchronous states.

\begin{definition}[Synchronous states]\label{def:synchronous-states}
A state $(x,o)\in\overline{\mathbf D}$ is called \emph{synchronous} if all
communities share a common infection level and a common shifted opinion, i.e.,
\[
x=\xi\mathbf{1}_n,\quad o=(\eta-0.5)\mathbf{1}_n
\]
for some $(\xi,\eta)\in[0,1]^2$. The set of all states is denoted by
\[
\mathcal S \coloneqq \{(x,o): x=\xi\mathbf{1}_n,\ o=(\eta-0.5)\mathbf{1}_n,\
(\xi,\eta)\in[0,1]^2\}.
\]
\end{definition}

\begin{definition}[Stability relative to an invariant set]
\label{def:relative-stability}
Let $\mathcal M$ be a positively invariant set for a dynamical system and let
$z^*\in\mathcal M$ be an equilibrium. The equilibrium $z^*$ is
\emph{stable relative to $\mathcal M$} if, for every $\varepsilon>0$, there
exists $\delta>0$ such that every solution initialized at
$z(0)\in\mathcal M$ with $\|z(0)-z^*\|<\delta$ satisfies
$\|z(t)-z^*\|<\varepsilon$ for all $t\ge0$. It is \emph{locally
asymptotically stable relative to $\mathcal M$} if it is stable relative to
$\mathcal M$ and solutions in a neighborhood of $z^*$ within $\mathcal M$
converge to $z^*$. It is \emph{locally exponentially stable relative to
$\mathcal M$} if there exist $c,\lambda,r>0$ such that
\[
\|z(t)-z^*\|\le c e^{-\lambda t}\|z(0)-z^*\|
\]
for all $t\ge0$ whenever $z(0)\in\mathcal M$ and
$\|z(0)-z^*\|<r$.
\end{definition}

\begin{assumption}\label{ass:homogeneous-benchmark}
The following conditions are satisfied:
\begin{enumerate}[label=(\roman*)]
    \item $\Gamma = \bar\gamma I_n$ for some $\bar\gamma\ge \gamma^{\min}$,
    and $\beta_1 = \bar\beta_1 I_n$, $\beta_2 = \bar\beta_2 I_n$ for some
    $\bar\beta_1\ge \beta_1^{\min}\ge 0$ and
    $\bar\beta_2\ge \beta_2^{\min}\ge 0$.
    \item $A\mathbf{1}_n = a\mathbf{1}_n$ for some $a\ge 0$.
    \item There exists $b\ge0$ such that
$\mathbf1_n^\top B_i\mathbf1_n=b, \; \forall i\in[n]$.
\end{enumerate}
\end{assumption}

Assumptions~\ref{ass:homogeneous-benchmark}(ii)--(iii) impose equal total
pairwise and quadratic higher-order input across communities. Consequently, a
synchronous state generates the same infection pressure in every community.
These conditions define a homogeneous setting under which $\mathcal S$ is
invariant; they do not imply that heterogeneous initial conditions generically
synchronize.

\begin{lemma}[Positive invariance of the synchronous set and reduced dynamics]
\label{lem:synchronous-manifold}
Under Assumption~\ref{ass:homogeneous-benchmark}, the set $\mathcal S$ is
positively invariant for system~\eqref{eq:HOI-SIOS}. Moreover, for every
initial condition $(x(0),o(0))\in\mathcal S$, every solution
$(x(t),o(t))$ of~\eqref{eq:HOI-SIOS} remains in $\mathcal S$ for all
$t\ge 0$, and the coordinates
$(\xi(t),\eta(t))$ defined by $x(t)=\xi(t)\mathbf{1}_n$ and
$o(t)=(\eta(t)-0.5)\mathbf{1}_n$ satisfy the planar system
\begin{equation}\label{eq:reduced-system}
\dot\xi=\xi\,\psi(\xi,\eta),\quad \dot\eta=\xi-\eta,
\end{equation}
where
\begin{equation}\label{eq:psi-def}
\psi(\xi,\eta)
= -d(\eta)
+ (1-\xi)\big(a\,\tilde\beta_1(\eta)
+ b\,\tilde\beta_2(\eta)\,\xi\big),
\end{equation}
with
\[
d(\eta)=\gamma^{\min}+(\bar\gamma-\gamma^{\min})\eta,
\]
\[
\tilde\beta_1(\eta)
=\bar\beta_1-(\bar\beta_1-\beta_1^{\min})\eta,
\quad
\tilde\beta_2(\eta)
=\bar\beta_2-(\bar\beta_2-\beta_2^{\min})\eta.
\]
\end{lemma}

\begin{proof}
By Definition~\ref{def:synchronous-states}, any $(x,o)\in\mathcal S$ can be
written as $x=\xi\mathbf{1}_n$ and $o=(\eta-0.5)\mathbf{1}_n$ for some
$(\xi,\eta)\in[0,1]^2$. Under
Assumption~\ref{ass:homogeneous-benchmark}, all diagonal rates become scalar
multiples of the identity:
\[
\mathcal D(o)=d(\eta)I_n,\quad
\mathcal B_q(o)=\tilde\beta_q(\eta)I_n,\quad q\in\{1,2\}.
\]
Also, $Ax=A(\xi\mathbf{1}_n)=a\xi\mathbf{1}_n$ by
Assumption~\ref{ass:homogeneous-benchmark} (ii). By Assumption~\ref{ass:homogeneous-benchmark}(iii),
\[
[\mathcal A_2(x)]_i
=x^\top B_i x
=\xi^2\mathbf1_n^\top B_i\mathbf1_n
=b\xi^2
\]
for every $i\in[n]$. Hence,
$\mathcal A_2(x)=b\xi^2\mathbf1_n$.
Substituting into~\eqref{eq:x_dynamics} with $X=\xi I_n$ gives
\[
\dot x
=\xi\Big(-d(\eta)+(1-\xi)\big(a\,\tilde\beta_1(\eta)+b\,\tilde\beta_2(\eta)\xi\big)\Big)\mathbf{1}_n
= \xi\,\psi(\xi,\eta)\,\mathbf{1}_n.
\]
For the opinion dynamics, since $o=(\eta-0.5)\mathbf{1}_n$, we have
$\Phi(o)=\sigma I_n$ with $\sigma=\sgn_m(\eta-0.5)\in\{-1,1\}$.
Therefore $\Phi(o)\bar L_u\Phi(o)o=(\eta-0.5)\bar L_u\mathbf{1}_n=0$, and
\[
\dot o=\xi\mathbf{1}_n-(\eta-0.5)\mathbf{1}_n-0.5\mathbf{1}_n=(\xi-\eta)\mathbf{1}_n.
\]
Since $(\dot x,\dot o)=(\alpha\mathbf{1}_n,\beta\mathbf{1}_n)$ for some scalars
$\alpha,\beta\in\mathbb R$, the vector field is tangent to $\mathcal S$.
It remains to be verified that $(\xi,\eta)$ stays in $[0,1]^2$:
\[
\xi=0 \Rightarrow \dot\xi=0,\quad
\xi=1 \Rightarrow \dot\xi=-d(\eta)\le -\gamma^{\min}<0,
\]
\[
\eta=0 \Rightarrow \dot\eta=\xi\ge 0,\quad
\eta=1 \Rightarrow \dot\eta=\xi-1\le 0.
\]
Hence, $\mathcal S$ is positively invariant. Identifying
$\dot\xi=\xi\psi(\xi,\eta)$ and $\dot\eta=\xi-\eta$ from the expressions
above yields~\eqref{eq:reduced-system}.
\end{proof}

At any equilibrium
of~\eqref{eq:reduced-system}, $\dot\eta=\xi-\eta=0$, so $\eta^*=\xi^*$.
Thus, either $\xi^*=0$ or $\psi(\xi^*,\xi^*)=0$. Define
\begin{equation}\label{eq:g-def}
g(s)\coloneqq \psi(s,s)
= -d(s)
+ (1-s)\big(a\,\tilde\beta_1(s)
+ b\,\tilde\beta_2(s)\,s\big),
\end{equation}
for $s\in[0,1]$. Hence, the positive equilibria
of~\eqref{eq:reduced-system} are the points $(s,s)$ satisfying
$g(s)=0$, and these correspond to equilibria of~\eqref{eq:HOI-SIOS} of the
form $(s\mathbf{1}_n,(s-0.5)\mathbf{1}_n)\in\mathcal S$.
Under the condition $a\bar\beta_1<\gamma^{\min}$ used below, we have
$g(0)<0$, while $g(1)<0$. Therefore, $g(\hat s)>0$ at some intermediate
prevalence implies at least two zero crossings in $(0,1)$. Thus, $g(\hat s)>0$ provides a directly checkable condition for multiple
positive equilibria on $\mathcal S$; it does not assume synchronization from
arbitrary initial conditions.
We now state the main bistability result for system~\eqref{eq:HOI-SIOS}
relative to the synchronous invariant set $\mathcal S$.
\begin{proposition}[Pairwise-only case]
\label{prop:pairwise-only-benchmark}
Suppose Assumption~\ref{ass:homogeneous-benchmark} holds and
$a\bar\beta_1<\gamma^{\min}$. If the higher-order transmission term is absent,
i.e., $b\,\tilde\beta_2(s)\equiv0$, then $g(s)<0$ for every $s\in[0,1]$.
Consequently, the reduced system~\eqref{eq:reduced-system} has no endemic
equilibrium and cannot exhibit healthy--endemic bistability.
\end{proposition}

\begin{proof}
When $b\,\tilde\beta_2(s)\equiv0$,
\(
g(s)=-d(s)+(1-s)a\,\tilde\beta_1(s).
\)
Since $d(s)\ge\gamma^{\min}$,
$\tilde\beta_1(s)\le\bar\beta_1$, and $1-s\le1$, we obtain
\[
g(s)\le-\gamma^{\min}+a\bar\beta_1<0
\]
for every $s\in[0,1]$. Hence, $g$ has no positive root.
\end{proof}

\begin{theorem}[Bistability relative to the synchronous set]
\label{thm:reduced-bistability}
Consider system~\eqref{eq:HOI-SIOS} under
\cref{ass:parameters,ass:network,ass:homogeneous-benchmark}, and let $g$ be
defined by~\eqref{eq:g-def}. Suppose that
\begin{equation}
\label{eq:healthy-reduced-condition}
a\bar\beta_1<\gamma^{\min},
\end{equation}
and that there exists $\hat s\in(0,1)$ such that $g(\hat s)>0$. Then the
following statements hold:
\begin{enumerate}[label=(\roman*)]
    \item system~\eqref{eq:HOI-SIOS} admits at least two endemic equilibria
    in $\mathcal S$;
    \item the healthy equilibrium
    $(\mathbf 0_n,-0.5\mathbf 1_n)$ is locally exponentially stable relative
    to $\mathcal S$.
\end{enumerate}
Suppose, in addition, that $g$ has exactly two simple roots
$0<s_-<s_+<1$ and that
\begin{equation}
\label{eq:large-root-trace-condition}
s_+\,\partial_\xi\psi(s_+,s_+)<1.
\end{equation}
Then system~\eqref{eq:HOI-SIOS} exhibits bistability relative to $\mathcal S$:
\begin{enumerate}[label=(\roman*)]
\setcounter{enumi}{2}
    \item the endemic equilibrium
    $\bigl(s_-\mathbf 1_n,(s_--0.5)\mathbf 1_n\bigr)$ is unstable relative
    to $\mathcal S$;
    \item the endemic equilibrium
    $\bigl(s_+\mathbf 1_n,(s_+-0.5)\mathbf 1_n\bigr)$ is locally
    exponentially stable relative to $\mathcal S$.
\end{enumerate}
\end{theorem}

\begin{proof}
By Lemma~\ref{lem:synchronous-manifold}, every solution initialized in
$\mathcal S$ remains in $\mathcal S$ and is represented by the reduced
system~\eqref{eq:reduced-system}. At an equilibrium of the reduced system,
$\dot\eta=0$ implies $\eta^*=\xi^*$. Hence, its positive equilibria are
exactly the points $(s,s)$ with $s\in(0,1)$ satisfying $g(s)=0$.

By~\eqref{eq:healthy-reduced-condition},
\(
g(0)=a\bar\beta_1-\gamma^{\min}<0,
\)
while $g(1)=-\bar\gamma<0$. Since $g(\hat s)>0$ for some
$\hat s\in(0,1)$, the intermediate value theorem implies that $g$ has at
least two roots in $(0,1)$. Each root gives an endemic equilibrium
$\bigl(s\mathbf 1_n,(s-0.5)\mathbf 1_n\bigr)\in\mathcal S$.
At $(\xi,\eta)=(0,0)$, the Jacobian of the reduced system is
\[
J_r(0,0)=
\begin{bmatrix}
a\bar\beta_1-\gamma^{\min} & 0\\
1 & -1
\end{bmatrix}.
\]
Its eigenvalues are $a\bar\beta_1-\gamma^{\min}<0$ and $-1<0$.
Therefore, $(\mathbf 0_n,-0.5\mathbf 1_n)$ is locally exponentially stable
relative to $\mathcal S$.

Now suppose that $g$ has exactly two simple roots
$0<s_-<s_+<1$. At a positive equilibrium $(s,s)$, the reduced Jacobian is
\[
J_r(s,s)=
\begin{bmatrix}
s\,\partial_\xi\psi(s,s) & s\,\partial_\eta\psi(s,s)\\
1 & -1
\end{bmatrix},
\]
and
\begin{equation}
\label{eq:reduced-det-trace}
\det J_r(s,s)=-s g'(s),
\qquad
\operatorname{tr}J_r(s,s)=s\,\partial_\xi\psi(s,s)-1.
\end{equation}
Since $g(0)<0$ and $s_-$ is the first simple zero crossing,
$g'(s_-)>0$. Thus, $\det J_r(s_-,s_-)<0$, so $(s_-,s_-)$ is a hyperbolic
saddle and the corresponding endemic equilibrium is unstable relative to
$\mathcal S$.

Since $s_+$ is the second simple zero crossing and $g(1)<0$, we have
$g'(s_+)<0$. Therefore, $\det J_r(s_+,s_+)>0$. Moreover,
condition~\eqref{eq:large-root-trace-condition} gives
$\operatorname{tr}J_r(s_+,s_+)<0$. Hence, $J_r(s_+,s_+)$ is Hurwitz.
Since the reduced vector field is smooth, Lyapunov's indirect method
implies that $(s_+,s_+)$, and thus
$\bigl(s_+\mathbf 1_n,(s_+-0.5)\mathbf 1_n\bigr)$, is locally
exponentially stable relative to $\mathcal S$.
\end{proof}
The quadratic higher-order term vanishes from the disease-free linearization,
so the healthy equilibrium can remain locally stable when pairwise transmission
is subcritical. At positive infection levels, however, higher-order transmission
can support a stable endemic equilibrium on $\mathcal S$, yielding bistability.
These conclusions apply to trajectories in $\mathcal S$ and do not imply
synchronization from arbitrary initial conditions.



\subsection{Endemic Regime}
When $\mathcal R_{\min}>1$, the pairwise epidemic linearization is
supercritical even at the most concerned opinion profile. Under the homogeneous
setting, we first establish the existence of an endemic equilibrium on
$\mathcal S$. For the general heterogeneous model, we then prove that every
healthy equilibrium away from the opinion switching surfaces is unstable.
These statements do not, by themselves, establish uniform persistence or
global stability of an endemic equilibrium.

\begin{proposition}[Existence of an endemic equilibrium on $\mathcal S$ when $\mathcal R_{\min}>1$]
\label{cor:Rmin-endemic-S}
Suppose \cref{ass:parameters,ass:network,ass:homogeneous-benchmark} hold. If
\(
\mathcal R_{\min}>1,
\)
then system~\eqref{eq:HOI-SIOS} admits at least one endemic equilibrium lying
in $\mathcal S$.
\end{proposition}

\begin{proof}
Under Assumption~\ref{ass:homogeneous-benchmark}, we have
\(
\Gamma=\bar\gamma I_n,
\;
B_1^{\min}=\beta_1^{\min}I_n,
\text{ and }
A\mathbf{1}_n=a\mathbf{1}_n.
\)
Since $A\mathbf{1}_n=a\mathbf{1}_n$ and $A\ge 0$, it follows that $\rho(A)=a$.
Thus,
\[
\mathcal R_{\min}
=
\rho(\Gamma^{-1}B_1^{\min}A)
=
\frac{\beta_1^{\min}}{\bar\gamma}\rho(A)
=
\frac{a\beta_1^{\min}}{\bar\gamma}.
\]
Hence, $\mathcal R_{\min}>1$ implies
\(
a\beta_1^{\min}>\bar\gamma.
\)

Now consider the function
\(g(s)\) defined in~\eqref{eq:g-def}.
At $s=0$,
\[
g(0)=a\bar\beta_1-\gamma^{\min}
\ge a\beta_1^{\min}-\bar\gamma
>0,
\]
since $\bar\beta_1\ge \beta_1^{\min}$ and $\gamma^{\min}\le \bar\gamma$.
Also,
\[
g(1)=-d(1)=-\bar\gamma<0.
\]
By the continuity of $g$, there exists $s^*\in(0,1)$ such that $g(s^*)=0$.
By Lemma~\ref{lem:synchronous-manifold}, this yields an equilibrium of
system~\eqref{eq:HOI-SIOS} of the form
\[
(x^*,o^*)=\big(s^*\mathbf{1}_n,(s^*-0.5)\mathbf{1}_n\big)\in\mathcal S.
\]
Since $s^*\in(0,1)$, we have $x^*=s^*\mathbf{1}_n\gg \mathbf{0}_n$, so this equilibrium
is endemic.
\end{proof}
\begin{proposition}[Instability of healthy equilibria when $\mathcal R_{\min}>1$]
\label{prop:Rmin-instability}
Consider system~\eqref{eq:HOI-SIOS} under Assumptions~\ref{ass:parameters}
and~\ref{ass:network}. If $\mathcal R_{\min}>1$, then every healthy
equilibrium $(\mathbf{0}_n,o^*)\in\overline{\mathbf D}$ satisfying
\(o_i^*\neq 0\) for all \(i\in[n]\) is unstable.
\end{proposition}

\begin{proof}
Let $(\mathbf{0}_n,o^*)$ be a healthy equilibrium satisfying \(o_i^*\neq 0\) for all
\(i\in[n]\). Then there exists a neighborhood of \((\mathbf{0}_n,o^*)\) on which
\(\Phi(o)\) is constant. Hence, in that neighborhood,
system~\eqref{eq:HOI-SIOS} coincides with a smooth ODE, and Lyapunov’s indirect method applies.

Since the HOI term is quadratic in $x$, it vanishes at $x=\mathbf{0}_n$. Therefore, the
epidemic block of the Jacobian at $(\mathbf{0}_n,o^*)$ is
\[
J_{xx}(\mathbf{0}_n,o^*)=-\mathcal D(o^*)+\mathcal B_1(o^*)A,
\]
which is a Metzler matrix.
By Assumption~\ref{ass:parameters}, for every $o^*\in[-0.5,0.5]^n$,
\[
\mathcal D(o^*)\le \Gamma,
\quad
\mathcal B_1(o^*)\ge B_1^{\min},
\]
componentwise on the diagonal entries. Since $\mathcal D(o^*)$ and $\Gamma$ are
positive diagonal matrices, inversion reverses the order:
\[
\mathcal D(o^*)^{-1}\ge \Gamma^{-1}.
\]
Hence,
\[
\mathcal D(o^*)^{-1}\mathcal B_1(o^*)A
\ge
\Gamma^{-1}B_1^{\min}A
\ge 0,
\]
componentwise. By the monotonicity of the spectral radius for nonnegative
matrices,
\[
\rho\big(\mathcal D(o^*)^{-1}\mathcal B_1(o^*)A\big)
\ge
\rho\big(\Gamma^{-1}B_1^{\min}A\big)
=
\mathcal R_{\min}.
\]
Therefore, if $\mathcal R_{\min}>1$, then
\[
\mathbf s\!\left(
-\mathcal D(o^*)+\mathcal B_1(o^*)A
\right)>0.
\]
Moreover, $J_{xo}(\mathbf0_n,o^*)=0$, so the full Jacobian at
$(\mathbf0_n,o^*)$ is lower block triangular and contains the eigenvalues
of $J_{xx}(\mathbf0_n,o^*)$. Hence, the full Jacobian has an eigenvalue with
positive real part, and Lyapunov's indirect method implies that
$(\mathbf0_n,o^*)$ is unstable
\cite[Theorem~4.7(ii)]{khalil2002nonlinear}.
\end{proof}

\section{Simulations}\label{sec:simulations}

We consider a network of $n=5$ communities. The matrices $A$ and $\bar A_u$
are binary, with an entry equal to one when the corresponding communities are
neighbors, and we set $A=\bar A_u$. For each $i\in[n]$, the matrix $B_i$ is binary: distinct-index entries encode the included 2-simplices, while
$b_{iii}=1$ represents within-community quadratic exposure. We choose
$\gamma^{\min}=6$, $\beta_1^{\min}=0.15$, $\beta_2^{\min}=0.3$,
$\gamma_i=6.1$, $\beta_{i1}=0.28$, and $\beta_{i2}=0.4$ for every $i\in[n]$, i.e., homogeneous spreading.
For this parameter set, $\mathcal R_{\max}=0.20$ and
$\rho\!\left(\Gamma_{\min}^{-1}(\beta_1A+\beta_2C)\right)=0.93$.
Theorem~\ref{thm:global-stability-lyap} therefore guarantees exponential
eradication of the infection state, as illustrated in Fig.~\ref{fig:Theo1}
(top). The opinion trajectory in Fig.~\ref{fig:Theo1} (bottom) numerically
approaches a dissensus profile. 
The next example illustrates why $\mathcal R_{\max}<1$ alone does not imply
global eradication in the presence of higher-order transmission.

\begin{figure}
    \centering
    \includegraphics[width=\linewidth]{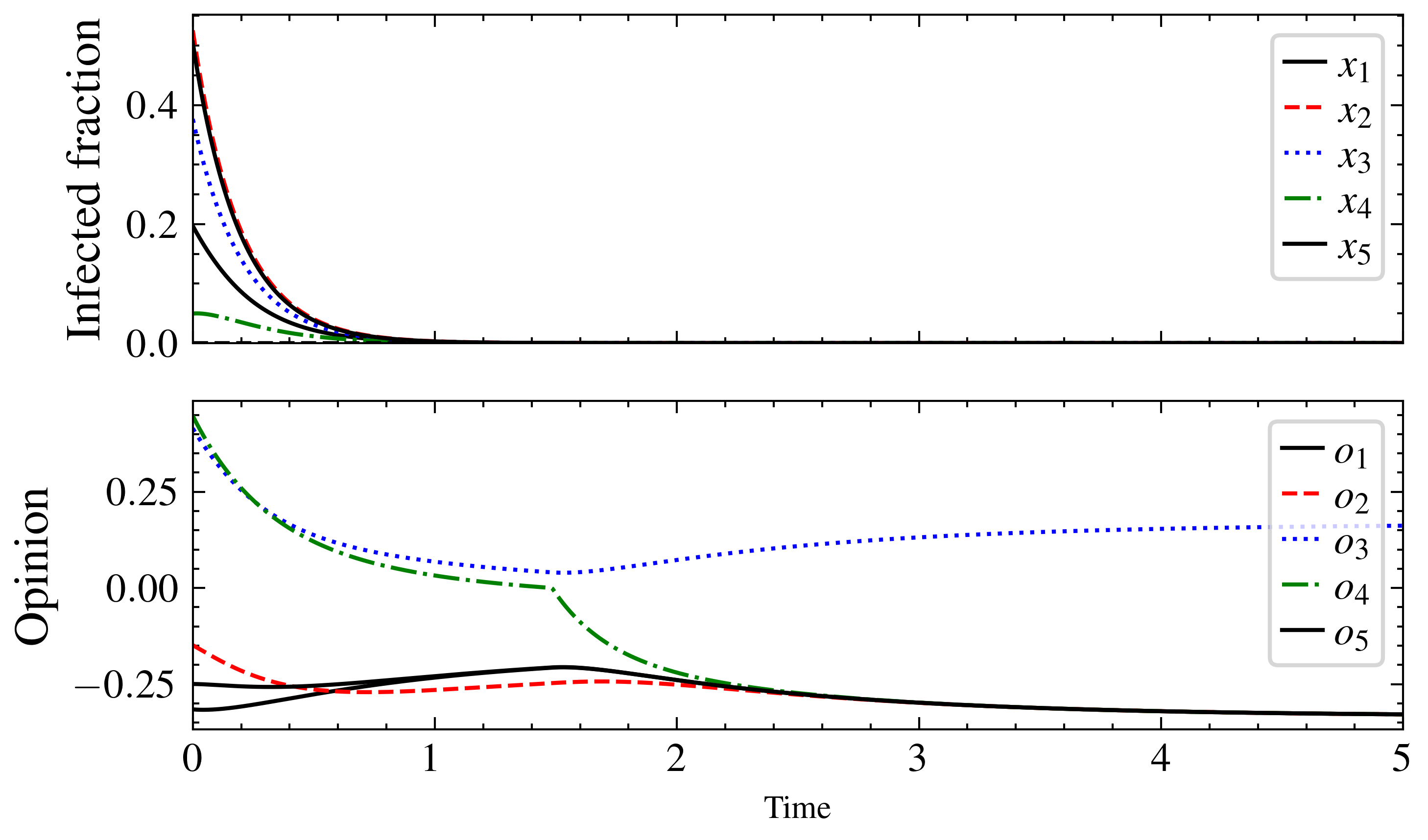}
    \caption{Evolution of the infection (top) and opinion (bottom)
dynamics when
$\rho\big(\Gamma_{\min}^{-1}(\beta_1A+\beta_2C)\big)<1$.
The infection state converges exponentially to $\mathbf{0}_n$, while the
opinion trajectory numerically approaches a dissensus profile.}
    \label{fig:Theo1}
\end{figure}

For the bistable example, choose
$\gamma^{\min}=1$, $\beta_1^{\min}=0.28$, $\beta_2^{\min}=0.4$,
$\gamma_i=1.1$, $\beta_{i1}=0.35$, and $\beta_{i2}=0.8$ for every $i\in[n]$.
The higher-order interaction matrices satisfy
$A\mathbf1_5=2\mathbf1_5$ and
$\mathbf1_5^\top B_i\mathbf1_5=5$ for every $i$, so Assumption~\ref{ass:homogeneous-benchmark}
holds with $a=2$ and $b=5$. Thus $\mathcal R_{\max}=0.7$, while
$\rho\!\left(\Gamma_{\min}^{-1}(\beta_1A+\beta_2C)\right)=20.7$.
We choose identical infection and opinion initial conditions across communities,
so the trajectories start in $\mathcal S$. 
For these parameters, $g$ has exactly two simple roots in $(0,1)$,
$s_-=0.1281$ and $s_+=0.5113$. Moreover,
\[
s_+\partial_\xi\psi(s_+,s_+)-1=-1.356<0.
\]
Thus, the conditions of Theorem~\ref{thm:reduced-bistability} are satisfied.
The healthy equilibrium $(\mathbf0_n,-0.5\mathbf1_n)$ and the larger endemic
equilibrium $(0.5113\mathbf1_n,0.0113\mathbf1_n)$ are locally exponentially
stable relative to $\mathcal S$, whereas the smaller endemic equilibrium
$(0.1281\mathbf1_n,-0.3719\mathbf1_n)$ is unstable relative to $\mathcal S$.
Figure~\ref{fig:Theo2} shows trajectories from selected initial conditions in
$\mathcal S$ that approach the healthy or larger endemic equilibrium,
consistent with Theorem~\ref{thm:reduced-bistability}.

\begin{figure}
    \centering
    \includegraphics[width=\linewidth]{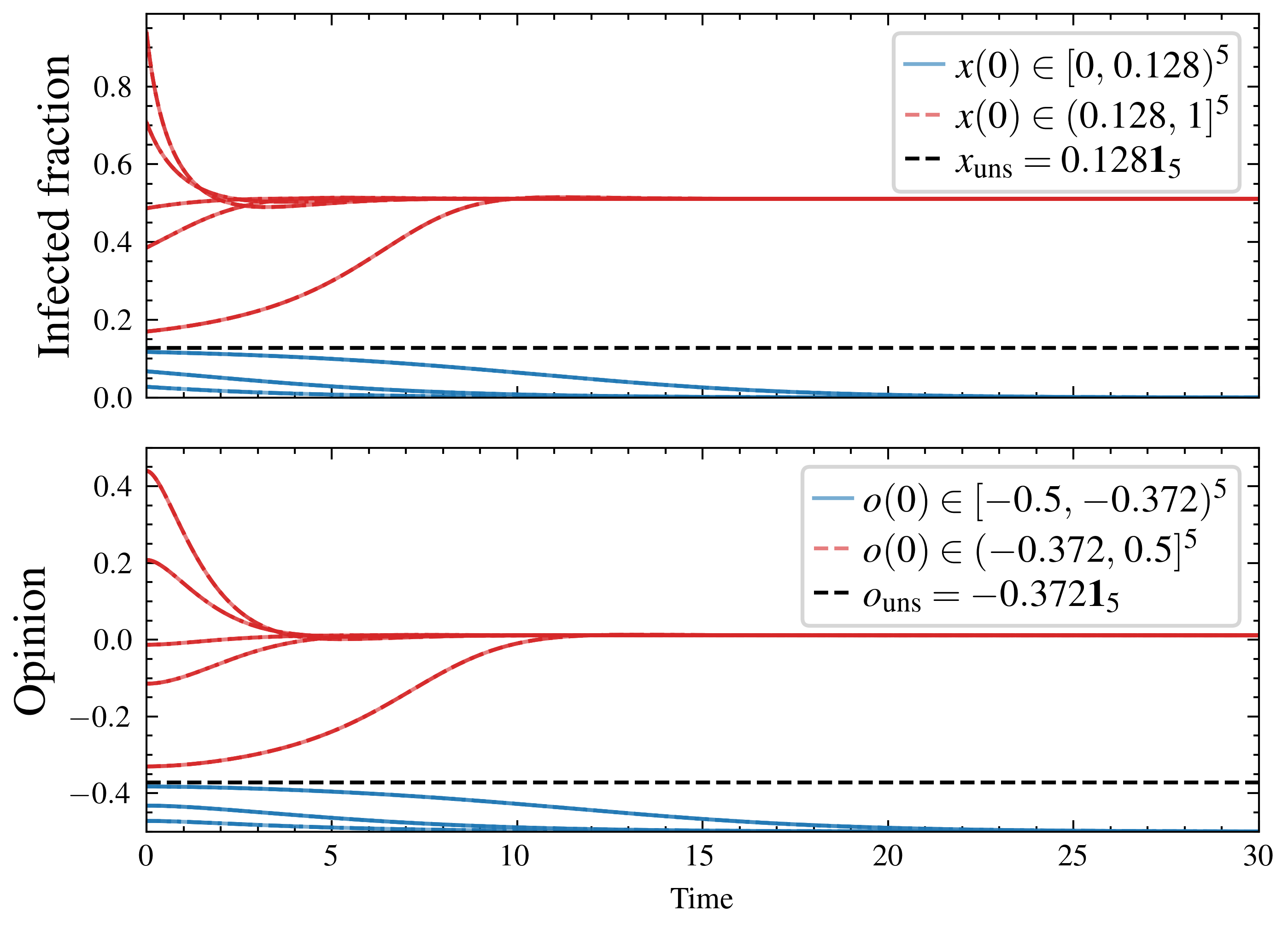}
    \caption{Selected trajectories of the synchronous reduced dynamics for the
bistable parameter set. Red trajectories approach the larger endemic
equilibrium $(0.5113\mathbf1_n,0.0113\mathbf1_n)$, whereas blue trajectories
approach the healthy equilibrium $(\mathbf0_n,-0.5\mathbf1_n)$. The black
dashed lines mark the smaller unstable endemic equilibrium
$(0.1281\mathbf1_n,-0.3719\mathbf1_n)$.}
    \label{fig:Theo2}
\end{figure}

Finally, we return to the first parameter set and change only the recovery rates
to $\gamma^{\min}=0.4$ and $\gamma_i=0.5$ for every $i\in[n]$. This gives
$\mathcal R_{\min}=1.10>1$, so
Proposition~\ref{prop:Rmin-instability} implies that every healthy equilibrium
away from the opinion switching surfaces is unstable. Figure~\ref{fig:Endemic}
shows one trajectory that numerically approaches an endemic-dissensus state.
This observation is consistent with the proposition, but it does not establish
global persistence or stability of the endemic state.

\begin{figure}
    \centering
    \includegraphics[width=\linewidth]{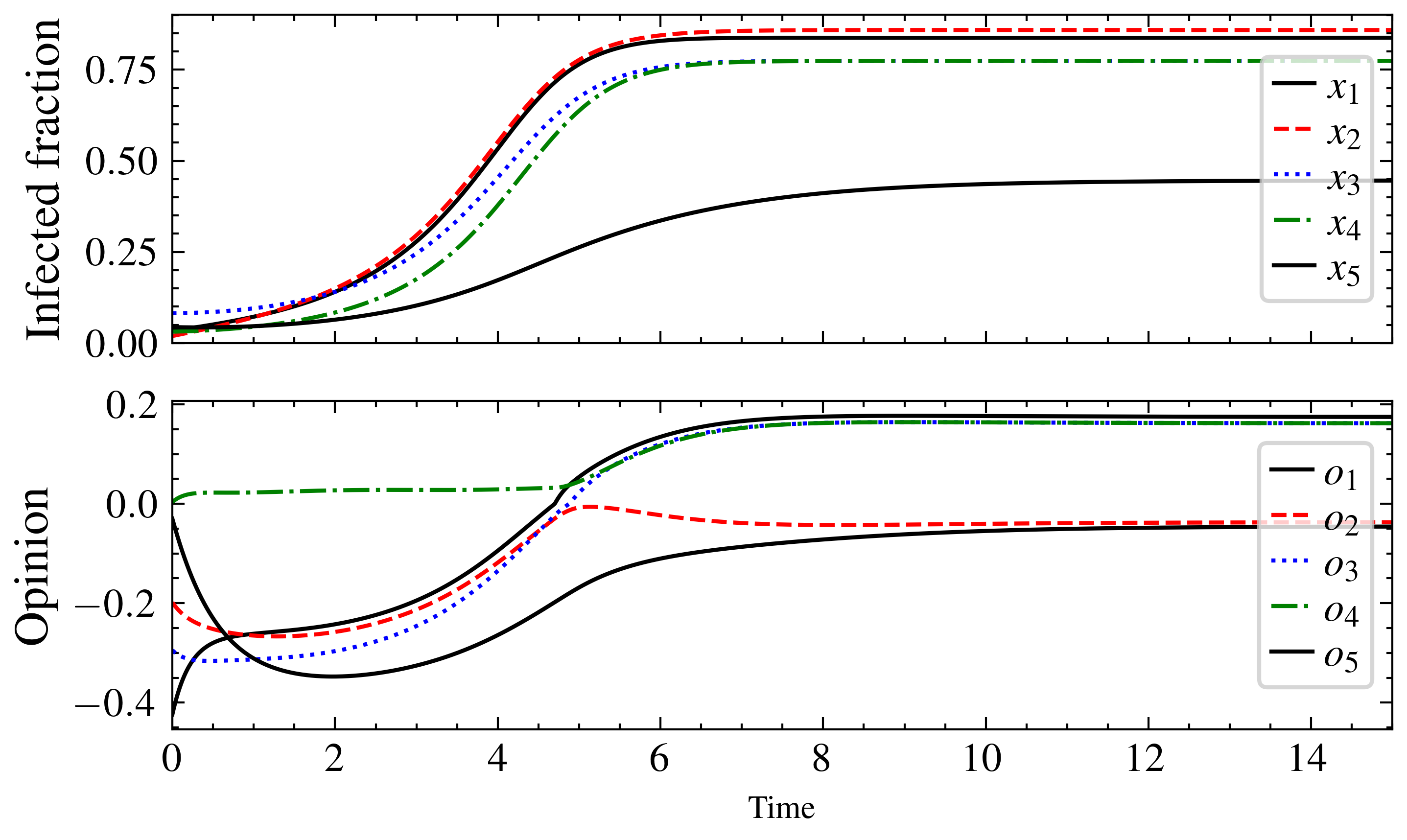}
    \caption{Trajectory approaching an endemic-dissensus state when
$\mathcal R_{\min}>1$. The instability of the healthy equilibrium
$(\mathbf0_n,-0.5\mathbf1_n)$ is consistent with
Proposition~\ref{prop:Rmin-instability}.}
    \label{fig:Endemic}
\end{figure}

\section{Conclusion}\label{sec:conclusion}
We introduced the HOI-SIOS model, in which epidemic prevalence and signed
opinion dynamics are coupled. We established local stability and instability conditions
for the maximally skeptical healthy equilibrium and a sufficient condition for
global exponential eradication of the infection state. Under a homogeneous
setting, the dynamics on a synchronous invariant set reduce to a planar
system and can exhibit healthy--endemic bistability. The quadratic higher-order term does not affect the disease-free linearization,
but it can sustain an endemic equilibrium at positive infection levels; the pairwise-only
case rules out this mechanism under the same healthy-stability condition.
We also showed that, when $\mathcal R_{\min}>1$, every healthy equilibrium away
from the opinion switching surfaces is unstable. Future work will examine
transverse stability beyond the synchronous set, convergence and switching in
the opinion subsystem, and the effects of time-scale separation and more general higher-order interactions.



\bibliographystyle{IEEEtran}
\bibliography{refs}

\end{document}